\documentclass[cmp,numbook,envcountsame,envcountreset]{svjour} 
\def\ctop{$C$-siblings\condblank}

\def\qed{\hfill\squareforqed}
\let\ssref=\ref
\def\fref#1{Fig.~\ssref{#1}}

\def\cref#1{Condition~\ssref{#1}}
\def\Cref#1{Corollary~\ssref{#1}}
\def\eref#1{(\ssref{#1})}
\def\sref#1{Sect.~\ssref{#1}}
\def\lref#1{Lemma~\ssref{#1}}
\def\rref#1{Remark~\ssref{#1}}
\def\tref#1{Theorem~\ssref{#1}}
\def\dref#1{Definition~\ssref{#1}}
\def\pref#1{Proposition~\ssref{#1}}
\def\exref#1{Example~\ssref{#1}}

\usepackage{times}
\usepackage{mhequ}
\usepackage{ScriptFonts}
\usepackage{Symbols}
\usepackage{subfig}
\usepackage{tikz}
\usetikzlibrary{snakes}

\def\thecomma{\ifx,\thenext \else\ifx;\thenext \else\ifx.\thenext
\else\ifx!\thenext \else\ifx:\thenext\else\ifx)\thenext \else \
\fi\fi\fi\fi\fi\fi}
\def\condblank{\futurelet\thenext\thecomma}
\def\ie{{\it i.e.}\condblank}

\newtheorem{assumption}[theorem]{Assumption}

\def\MM{\mathcal{M}}

\def\VVoldwd{\DD^{\,d}}

\def\WW{\mathcal{W}}

\def\FF{\mathcal{F}}
\def\GG{\mathcal{G}}

\def\NN{\mathcal{N}}
\def\DD{\mathcal{D}}
\def\BB{\mathcal{B}}
\def\TT{\mathcal{T}}
\def\SS{\mathcal{S}}
\def\JJ{\mathcal{J}}

\def\EE{\mathcal{E}}

\def\PP{\mathcal{P}}
\def\VV{\mathcal{V}}

\def\G{{\rm G}}

\def\P{{\rm {P}}}
\def\S{{\rm {S}}}

\def\1{{\bf 1}}
\def\argcdot{{\,\cdot\,}}

\let\phi=\varphi
\let\rho=\varrho

\def\real{{\bf R}}

\def \sumv{\sum_{v \in \TT }}
\def \sumvp{\sum_{v' \in \TT }}
\def \sumvd{\sum_{v \in \VVoldwd }}

\def\Label#1{}

\usepackage{stmaryrd}
\def\scomm#1#2{\ensuremath{\bigr\llbracket #1:#2\bigr\rrbracket}}

\begin{document}

\title{Controlling General Polynomial Networks}
\bigskip

\author{N. Cuneo\inst{1} and J.-P. Eckmann\inst{1}\fnmsep\inst{2}}
\institute{D\'epartement de Physique Th\'eorique, Universit\'e de
  Gen\`eve \and
Section de Math\'ematiques, Universit\'e de
  Gen\`eve}

\maketitle

\begin{abstract}
We consider networks of massive particles connected by
 non-linear springs. Some particles interact with heat baths at
different temperatures, which
are modeled as stochastic driving forces.
 The structure of the network is
arbitrary, but the motion of each particle is 1D.
For polynomial interactions, we give sufficient conditions
for H\"ormander's ``bracket condition''
to hold, which implies the
uniqueness of the steady state (if it exists), as well as the
controllability of the associated system in control theory. These
conditions are constructive; they are formulated in terms of
 inequivalence of the forces (modulo
translations) and/or conditions on the topology of  the
connections. We illustrate our results with
examples, including ``conducting chains'' of variable cross-section.
This then extends the results for a simple chain obtained in \cite{EPR1999} .
\end{abstract}

\tableofcontents
\section{Introduction}
We consider a network of interacting particles described
by an undirected graph $\GG = (\VV, \EE)$ with a set $\VV$ of vertices
 and a set $\EE$ of edges. Each vertex
represents a particle, and each edge represents a spring connecting
two particles. We single out a set $\VV_* \subset \VV$ of particles,
 each of which interacts with a heat bath.
We
address the question of when such a system has a unique stationary
state.  This question has been studied for several special cases:
Starting from a linear chain \cite{EPR1999,EPR1999b}, results have become
more refined in terms of the relation between the spring potentials
and the pinning potentials which tie the masses to the laboratory
frame \cite{LTh00,EH}.  This problem is very delicate, as is
apparent from the extensive study in \cite{Hairer2masses2009} for the
case of only 2 masses.

We provide conditions on the interaction potentials
that imply H\"ormander's ``bracket condition,'' from which it follows
that the semigroup associated to the process has a smoothing effect.
This, together with some stability assumptions, implies the
 \emph{uniqueness} of the stationary state. The {\em existence}
 is not discussed in this paper, but seems well
understood in the sense that the interaction potentials must be somehow
stronger than the pinning potentials. This issue will be explained in
a forthcoming paper \cite{EHR}.

Since the problem is known (see for example \cite{hairer_probabilistic_2005}
and \cite{bellet_ergodic_2006})
to be tightly related to the control problem where the stochastic
driving forces are replaced
with deterministic control forces, we shall use the
terminology of control theory,
and mention the implications of our results from the control-theoretic
viewpoint.

We work with unit masses and interaction
potentials that are polynomials of degree at least 3, and we say that two such
potentials  $V_1$ and $V_2$ have \emph{equivalent second derivative}  if there
is
a  $\delta\in \real$ such that
$V_1''(\argcdot) = V_2''(\argcdot+\delta)$.

We start with the set $\VV_*$ of particles
that interact with heat baths, and are therefore
{\em controllable}. One of our results (\Cref{c:aloneinequiv}) is formulated
as a condition  for some of the particles in the
set of first neighbors $\NN(\VV_*)$ of $\VV_*$ to be also controllable.
Basically, the condition is that
these particles must be ``inequivalent"
 in a sense that involves both the topology of their connections to
 $\VV_*$ and the corresponding interaction potentials. More precisely,
a sufficient condition for a particle $v\in \NN(\VV_*)$ to be controllable
is that for each other particle $w\in \NN(\VV_*)$ at least one of the
two conditions holds:
\begin{itemize}
\item[(a)] $v$ and $w$ are connected to $\VV_*$ in a topologically different
way,
\item[(b)] there is a particle $c$ in $\VV_*$ such that the
  interaction potential between $c$ and $v$ and that between $c$ and $w$
have  inequivalent second derivative.
\end{itemize}
It is then possible to use this condition recursively,
 taking control of more and more masses at each step (\tref{t:resultCk}).
If by doing so we can control all the masses in the graph, then
 H\"ormander's bracket condition holds.

In \sref{s:examples} we give examples of physically relevant networks
 whose controllability can be established using this method.

Our results imply in particular that connected graphs are controllable for
``almost all" choices
of the interaction potentials, provided that they
 are polynomials of degree at least 3 (\Cref{c:genericity}).

\section{The system}

We define a Hamiltonian for the graph  $\GG=(\VV, \EE)$ as follows.
Each particle $v \in \VV$ has position $q_v \in \real$
and momentum $p_v \in \real$ and is ``pinned down" by a potential
$U_v(q_v)$. Throughout, we assume the masses being 1, for simplicity
of notation.
See \rref{r:414} on how to adapt the results when the masses are not all equal.

We denote each edge $e\in \EE$ by $\{u, v\}$ (or equivalently $\{v, u\}$)
 where $u,v$ are the vertices adjacent to $e$.\footnote{Due to the physical
nature of the problem, we assume that the graph has no self-edge.} To each
edge $e=\{u, v\}$,
 we associate an interaction potential
$V_{uv}(q_u -q_v)$, or equivalently $V_{vu}(q_v-q_u)$ with
\begin{equ}\label{e:equivexpr}
V_{vu}(q_v-q_u) \equiv V_{uv}(q_u-q_v)~.
\end{equ}
Note that we do not require the potentials to be even functions; the
condition \eref{e:equivexpr}
just
 makes sure
that considering $e=\{u,v\}$ or $e=\{v, u\}$ leads to the same physical
interaction,
which is consistent with the fact that the edges are not oriented.

With the notation $q= (q_v)_{v\in \VV}$ and $p= (p_v)_{v\in \VV}$ the
Hamiltonian
 is then
\begin{equ}
  H(q, p)= \sum_{v\in \VV} \bigl( p_v^2/2 + U_v(q_v)\bigr) + \sum_{e \in
\EE}V_{e}(\delta q_e)~,
\end{equ}
where it is understood that $V_{e}(\delta q_e)$ denotes one of the
two expressions in \eref{e:equivexpr}.

Finally, we make the following  assumptions:
\begin{assumption}\label{assume}
\hfill\begin{enumerate}
\item All
functions are smooth. 
\item The level sets of $H$ are compact,  
\ie,  for each $K>0$ the set $\{(q, p)~|~H(q, p) \leq K\}$ is
compact. 
\item The function $\exp(-\beta H)$ is integrable for
some $\beta >0$.
\end{enumerate}
\end{assumption}
Each particle $v \in \VV_*$ is coupled to a heat bath at temperature $T_v >0$
with
coupling
 constant $\gamma_v >0$. For convenience, we set $\gamma_v = 0$ when
$v\notin \VV_*$.
The model is then described by the system of stochastic
differential equations
\begin{equa}\label{e:stoch}
  dq_v&= p_v\, dt~,\\\ dp_v &=-U_v'(q_v)dt-\partial_{q_v} \Bigl(\sum_{e \in
\EE}V_{e}(\delta q_e)\Bigr) dt -\gamma_v p_v dt
  +\sqrt{2T_v\gamma_v}\,d W_v(t)~,
\end{equa}
where the $W_v$ are identical independent Wiener processes.
The solutions to \eref{e:stoch}
form a Markov process. The generator of the associated semigroup is given by
\begin{equ}\Label{e:liouville}
L \equiv X_0 + \sum_{v\in \VV_*}\gamma_vT_v\partial_{p_v}^2~,
\end{equ}
with
\begin{equ}\Label{e:X0}
X_0 \equiv-\sum_{v\in \VV_*} \kern -0.2em\gamma_v p_v \partial_{p_v} + \sum_{v
\in \VV}\bigl({p_v
\partial_{q_v}}
-U'_v(q_v) \partial_{p_v}\bigr)  -\sum_{\{u,v\} \in \EE}\kern -0.5em
V'_{uv}(q_u-q_v)\cdot(\partial_{p_u}-\partial_{p_v})~.
\end{equ}
{}From now on, we assume that the interaction potentials $V_e$,
 $e\in \EE$ are \emph{polynomials of degree at least
3}. The condition on the degree means that we require throughout the presence
of non-harmonicities. The fully-harmonic case has been
described earlier \cite{EZ2004}, and the case where some but not all the
 potentials are harmonic is not covered here. We will show in a counter-example
(\exref{ex:anharmon}) that the non-harmonicities are really essential for our
results.
We make no assumption about the pinning potentials $U_v$; we do not require
them to be polynomials, and some of them may be identically
zero.

We work in the space $\real^{2|\VV|}$ with coordinates $x=(q, p)$.
 We identify the vector fields over $\real^{2|\VV|}$ and the corresponding
first-order differential operators in the usual way. This enables us to
consider
Lie algebras of vector fields over $\real^{2|\VV|}$, where the Lie bracket
$[\argcdot,\argcdot]$ is the usual commutator of two operators.

\begin{definition}We define $\MM$ as the smallest
Lie algebra that
\begin{itemize}
\item[(i)] contains $\partial_{p_v}$ for all $v\in \VV_*$,\footnote{Due to
the identification mentioned above, we view here $\partial_{p_v}$
 as a constant vector field over $\real^{2|\VV|}$.}
\item[(ii)] is closed under the operation $[\argcdot, X_0]$,
\item[(iii)] is closed under multiplication by smooth scalar functions.
\end{itemize}
By the definition
of a Lie algebra, $\MM$ is closed under
linear combinations and Lie brackets.
\end{definition}

\begin{definition}\label{d:controlnet}We say that a particle $v\in \VV$ is {\em
controllable}
if we have $\partial_{q_v}, \partial_{p_v} \in \MM$. We say that the network
$\GG$
 is {\em controllable} if all the particles are controllable, \ie, if
 \begin{equ}\label{e:touslespqdansM}
\partial_{q_v}, \partial_{p_v} \in \MM\qquad \text{ for all } v\in \VV~.
\end{equ}
\end{definition}

\emph{The aim of this paper is to give sufficient conditions on $\GG$
  and the interaction potentials, which guarantee that the network
  is controllable.}

If the network is controllable in the sense \eref{e:touslespqdansM}, then
H\"ormander's condition\footnote{The condition \eref{e:hormander} is slightly
different, but equivalent to the
usual statement of H\"ormander's criterion. This can be checked easily.
 In particular, closing $\MM$ under multiplication
by smooth scalar functions does
not alter the set in \eref{e:hormander}, and will be very convenient.}
\cite{Ho} holds: for all $x$, the vector fields $F\in\MM$
evaluated at $x$ span all of
$\real^{2|\VV|}$, \ie,
\begin{equ}\label{e:hormander}
\{F(x)~|~F\in \MM\} = \real^{2|\VV|}
 \quad \text{ for all } x\in \real^{2|\VV|}~.
\end{equ}

H\"ormander's condition implies that the transition probabilities of
the Markov process \eref{e:stoch}  are smooth,
 and that so is any invariant measure
(see for example \cite[Cor.~7.2]{bellet_ergodic_2006}).
We now briefly mention two implications of these smoothness properties.
\pref{p:propmeasure} and \pref{p:control}  below
 can be deduced from arguments similar to those exposed in
\cite{hairer_probabilistic_2005},
and will be discussed in more detail in the forthcoming paper \cite{EHR}.

\begin{proposition}\label{p:propmeasure} Under
  Assumption~\ref{assume}, 
if \eref{e:hormander} holds,
then the Markov
process \eref{e:stoch} has at most one
invariant probability measure.
\end{proposition}

The control-theoretic problem corresponding to \eref{e:stoch} is
the system of ordinary differential equations
\begin{equa}[1][e:controlsys]
  \dot{q}_v&= p_v,\\\ \dot{p}_v &=-U_v'(q_v)-\partial_{q_v} \Bigl(\sum_{e \in
\EE}V_{e}(\delta q_e)\Bigr)
  +( u_v(t)-\gamma_v p_v)\cdot{\bf 1}_{v\in \VV_*}~,
\end{equa}
where for each $v\in \VV_*$, $u_v: \real \to \real$ is a smooth {\em control
function} (\ie,
the stochastic driving forces have been replaced with deterministic
functions).\footnote{Whether
or not we keep the dissipative terms $-\gamma_v p_v$ in \eref{e:controlsys}
makes no difference
 since they can always be absorbed in the control functions.}

\begin{proposition}\label{p:control}Under the hypotheses of
\pref{p:propmeasure}, the system
 \eref{e:controlsys} is {\em controllable}
in the sense that for each $x^{(0)} = (q^{(0)}, p^{(0)})$ and $x^{(f)}= (q^{(f)}, p^{(f)})$,
there are a time $T$ and some smooth controls $u_v$, $v\in \VV_*$,
such that the solution $x(t) $ of \eref{e:controlsys} with $x(0) =x^{(0)}$ verifies
$x(T) =x^{(f)}$.
\end{proposition}

In fact, \eref{e:hormander} is a well-known condition in control theory.
See for example \cite{jurdjevic_geometric_1997}, which addresses the case
of piecewise constant control functions. In particular, \eref{e:hormander}
implies by
\cite[Thm.~3.3]{jurdjevic_geometric_1997} that for every initial condition
$x^{(0)}$
and each time $T>0$, the set $A(x^{(0)}, T)$ of all points that are accessible at
time $T$
(by choosing appropriate controls) is connected and full-dimensional.

\section{Strategy}\label{s:commutators}

We want to show that $\partial_{q_v},\partial_{p_v} \in \MM$ for all
$v\in \VV$.  The next lemma shows that we only need to worry about
 the $\partial_{p_v}$.

\begin{lemma}\label{l:piimpliesqi} Let  $A$ be a subset of\/ $\VV$.
\begin{equ}
\text{If }~~\sum_{v\in A} \partial_{p_v} \in \MM ~~\text{ then }~~
\sum_{v\in A} \partial_{q_v} \in \MM~.
\end{equ}
\end{lemma}
\begin{proof} Assuming $\sum_{v\in A} \partial_{p_v}\in\MM$, we find that
\begin{equ}\label{e:sommevina}
\bigr[\sum_{v\in A} \partial_{p_v}, {X}_0\bigr] = \sum_{v\in A} \partial_{q_v}
-
\sum_{v\in \VV_*\cap A} \gamma_v \partial_{p_v}
\end{equ}
is in $\MM$. But since $\partial_{p_v} \in \MM$ for all $v\in \VV_*$,
 the linear structure of $\MM$ implies
 $\sum_{v\in \VV_*\cap A} \gamma_v \partial_{p_v}
 \in \MM$. Adding this to the vector field in \eref{e:sommevina}
shows that\hfill\break  $\sum_{v\in A} \partial_{q_v} \in \MM$, as claimed.
\qed\end{proof}

\begin{definition}
We say that a set $A \subset \VV$ is \emph{jointly controllable} if
$\sum_{v\in A}\partial_{p_v}$ is in  $\MM$ (and therefore, also
$\sum_{v\in A}\partial_{q_v}$ by \lref{l:piimpliesqi}).
\end{definition}

Requiring all the particles in a set $A$ to be (individually) controllable is
 stronger than asking the set $A$ to be jointly controllable (indeed, if all
the $\partial_{p_v}$,
$v\in A$ are in $\MM$, then so is their sum). We will
obtain
jointly controllable sets and then ``refine" them until we control particles
individually.

The strategy is as follows. In the next section, we start with a controllable
particle $c$,
and show that its first neighbors split into jointly
controllable sets. Then, in \sref{s:controlnet}, we consider several
controllable
particles, and basically intersect
the jointly controllable sets obtained for each of them, in order to control
``new" particles individually.
 Finally, we iterate this procedure, taking control of more particles at each
step,
 until we establish (under some conditions) the controllability of the whole
network.

\begin{remark}Observe in the following that our results neither involve the
pinning potentials $U_v$ nor the coupling constants $\gamma_v$.
\end{remark}

\section{The neighbors of one controllable particle}\label{s:neighcontroller}

We consider in this section a particle $c \in \VV$, and denote by $\TT^c$ the
set of its first neighbors (the ``targets").
 The following notion of equivalence
among polynomials enables us to split $\TT^c$ into equivalence
classes.

\begin{definition}
Two polynomials $f$ and $g$ are called \emph{equivalent} if there
is a $\delta \in \real$
such that $f(\argcdot)=g(\argcdot+\delta )$.
\end{definition}

\begin{definition}
We say that two particles $v, u\in \TT^c$ are equivalent (with respect to $c$)
if the two polynomials $V''_{cv}$ and $V''_{cu}$ are equivalent.
\end{definition}

Since this
relation is symmetric and transitive, the set $\TT^c$ is naturally
decomposed into a disjoint union of equivalence classes:
\begin{equ}\Label{e:tti}
  \TT^c=\cup_i \TT^c_i~.
\end{equ}
An explanation of why we use the second derivative of the potentials
instead of the first one (\ie, the force) will be given in
\exref{ex:limitesderivee}.
The main result of this section is
\begin{theorem}\label{t:mainprop}
Assume that $c$ is controllable. Then, each equivalence class $\TT^c_i$ is
jointly
controllable, \ie,
\begin{equ}\label{e:vectorfield}
  \sum_{v\in\TT^c_i} \partial_{p_v} \in \MM~\quad \text{ for all }i~.
\end{equ}
Furthermore, there are constants $\delta _{cv}$ such that for all $i$,
\begin{equ}\label{e:vectorfieldx}
  \sum_{v\in\TT^c_i} (q_c-q_v+\delta _{cv})\partial_{p_v} \in \MM~.
\end{equ}
\end{theorem}

The second part of the theorem will be used in the next section to
intersect
 the equivalence classes $\TT^c_i$ of several
 controllable particles $c$.
We will now prepare the proof of \tref{t:mainprop}. We assume in the remainder
of
 this section that $c$ is controllable. And since $c$
is fixed, we write $\TT$ and $\TT_i$ instead of $\TT^c$ and $\TT^c_i$.

\begin{lemma}\Label{l:startingvectorsinM}
We have
\begin{equ}\label{e:partial}
\sumv V''_{cv}(q_c-q_v)\partial_{p_v}\in \MM~.
\end{equ}
\end{lemma}

\begin{proof} From \lref{l:piimpliesqi}
we conclude that
  $\partial_{q_c} \in \MM$. Therefore, we find that
\begin{equ}\Label{e:first}
[\partial_{q_c}, {X}_0] = -U_c''(q_c)\partial_{p_c} -  \sumv
V''_{cv}(q_c-q_v)\cdot(\partial_{p_c}-\partial_{p_v})
\end{equ}
is in $\MM$. Now, since $\partial_{p_c} \in \MM$ and since $\MM$ is closed
 under multiplication by scalar functions, we can subtract all the
contributions
that are along $\partial_{p_c}$ and obtain \eref{e:partial}.
\qed\end{proof}

We need a bit of technology to deal
with equivalent polynomials.
\begin{definition}Let $g(t) = \sum_{i=0}^k a_i {t^i / i!}$ be a polynomial
of degree $k\geq 1$. If $a_{k-1} = 0$, we say that $g$ is {\em adjusted}.
As can be checked, the polynomial $\tilde g(\argcdot)\equiv g(\argcdot
-a_{k-1}/a_k)$
 is always adjusted, and is referred to as the {\em adjusted representation} of
$g$.\end{definition}

Observe that a polynomial and its adjusted representation are by construction
equivalent and have the same degree and the same leading coefficient. In
fact, given a polynomial $g$ of degree $k\geq 1$, $\tilde{g}$ is the only
 polynomial equivalent to $g$ that is adjusted.
This adjusted representation will prove to be very
useful thanks to the following obvious
\begin{lemma}\label{l:45}
Two polynomials $g$ and $h$ of degree at least 1
 are equivalent iff\, $\tilde{g} =\tilde{h}$.
\end{lemma}

\begin{remark}
 If all the interaction potentials are \emph{even}, then all the $V_{cv}''$ are
 automatically adjusted, and some parts of the following discussion can be simplified.
\end{remark}

We shift the argument of each
$V_{cv}''$ by a constant $\delta _{cv}$ so that they are all adjusted.
We let $\tilde f_v$ be the adjusted representation of $V''_{cv}$ and
use the notation
\begin{equ}
x_v \equiv q_{c}-q_v + \delta_{cv}
\end{equ}
so that
\begin{equ}\Label{e:ftilde}
 \tilde{f}_v(x_v)=V''_{cv}(q_c-q_v) \quad \text{ for all }\quad  q_c,
q_v\in \real~.
\end{equ}
With this notation, \eref{e:partial} reads as
\begin{equ}\label{e:newnotation}
  \sum_{v\in\TT}  \tilde f_v(x_v) \partial _{p_v}\in\MM~.
\end{equ}
We will now mostly deal with ``diagonal" vector fields, \ie, vector fields
of the kind
\eref{e:newnotation}, where the component along $\partial_{p_v}$
 depends only on
$x_v$. When taking commutators, it is crucial to remember that
$x_v$ is only a notation for $ q_{c}-q_v + \delta_{cv}$.

\begin{remark}\label{r:ftildediff} By the definition of equivalence and
\lref{l:45},
 two particles $v,w\in \TT$ are equivalent iff $\tilde{f}_v$ and $\tilde{f}_w$
coincide.
\end{remark}

\begin{lemma}\label{l:derrivecomponentwise}
Consider some functions $g_v$, $v\in\TT$.
\begin{equ}\label{e:implicationderivee}
\text{ If }~~\sumv  g_v(x_v)\partial _{p_v}\in\MM \qquad\text{ then }\qquad
\sumv  g_v'(x_v)\partial _{p_v}\in\MM~.
\end{equ}
\end{lemma}
\begin{proof}
  This is immediate by commuting with $\partial _{q_c}$ (which is in
  $\MM$ by \lref{l:piimpliesqi}).
\qed\end{proof}

We now introduce the main tool.

\begin{definition}
Given two vector fields $Y$ and $Z$, we define
 the {\em double commutator} \scomm{Y}{Z} by
\begin{equ}~
\scomm{Y}{Z}\equiv [[X_0,Y], Z]~.
\end{equ}
\end{definition}
Obviously, if the vector fields $Y$ and $Z$ are in $ \MM$,
 then so is \scomm{Y}{Z}.

\begin{lemma}\label{l:doublecomm}
Consider some functions $g_v$, $h_v$,
$v\in\TT$.
Then
  \begin{equ}\label{e:scommvecteursdiagxu}
\scomm{\sumv g_v(x_v)\partial_{p_v}}{\sumvp
 h_{v'}(x_{v'})\partial_{p_{v'}}}= \sumv (g_v h_v)'(x_v)\,\partial _{p_v}~.
\end{equ}
\end{lemma}
\begin{proof}
Observe first that (omitting the arguments $x_v$)
  \begin{equa}\null
 \bigl  [\,{X}_0, \sumv& g_v \partial_{p_v} \,\bigr]
=\sumv  (p_c-p_v)
    g_v'\partial_{p_v}-\sumv
    g_v\partial _{q_v}+\sum_{v\in\TT\cap \VV_*}\kern -0.5 em \gamma_v
g_v\partial _{p_v}~.
  \end{equa}
Commuting with $\sumvp
 h_{v'}(x_{v'})\partial_{p_{v'}}
$ gives the desired result.
\qed\end{proof}

We will prove \tref{t:mainprop} starting from
\eref{e:newnotation} and using only  \eref{e:implicationderivee} and
 double commutators of the kind \eref{e:scommvecteursdiagxu}.

Let $d_v$ be the degree of  $\tilde f_v$. Note that since the
interaction potentials are of degree at least 3, we have $d_v\ge1$. We define
\begin{equ}
d \equiv \max_{v \in \TT } d_v \ge 1
\end{equ}
as the maximal degree of the adjusted polynomials $\tilde f_v$ with $v\in \TT$.
We can then write
\begin{equ}\Label{e:fuscoeffsai}
\tilde{f}_v(x) = \sum_{j=0}^{d} {b_{vj}}\,x^j/j!~,
\end{equ}
 for some real coefficients  $b_{vj}$, $j=0, \dots, d$, with
\begin{equ}
b_{vj} = 0 \quad \text { if } \quad j > d_v \quad \text { and }
\quad
b_{v,d_v-1} = 0~,
\end{equ}
for all $v\in \TT$.

\begin{definition}
We define the set of particles $v \in \TT$ corresponding to the
maximal
degree $d$:
\begin{equ}
\VVoldwd  \equiv \{v \in \TT  ~|~ d_v = d\}~.
\end{equ}
For every $\ell$, $0\leq \ell \leq d$, we define the set
\begin{equ}
\BB_\ell^{\,d} \equiv \{b_{v\ell} ~|~v\in \VVoldwd \}
\end{equ}
of {\em distinct} values taken by the coefficients
of ${x_v^\ell/\ell!}$ in $\tilde{f}_v$~, $v\in \VVoldwd$.
\end{definition}
We begin with a technical lemma. Observe how it is expressed
in terms of the $x_v$. In a sense, this shows that the $x_v$ are really the
 ``natural" variables for this problem. Thus, in addition to making the notion
of equivalence
trivial (\rref{r:ftildediff}), working with adjusted representations will be
very convenient
from a technical point of view.

\begin{lemma}\Label{l:47}
The following hold:
\begin{itemize}
\item[(i)] For each  $b\in \BB^{\,d}_d $\,, we have
  \begin{equ}\label{e:spansuraxeta}
    \sum_{v\in \VVoldwd~:~ b_{vd}=b}x_v\partial_{p_v}
 \in \MM \quad \text{ and }
\sum_{v\in \VVoldwd~:~b_{vd}=b}\partial_{p_v}  \in \MM~.
  \end{equ}
\item[(ii)] Furthermore,
  \begin{equ}\label{e:fullvector1dansMh}
    \sumvd x_v\partial_{p_v}  \in \MM \quad \text{ and }
\quad \sumvd\partial_{p_v}  \in \MM~.
  \end{equ}
\item[(iii)] Let $\alpha_v, \beta_v$, $v\in \VVoldwd$ be real constants.
If $d\geq 2$, we have the two implications
\begin{equs}
\text{ if } \sumvd \alpha_v\ \partial_{p_v}\in \MM~,
\quad &\text{ then }~~ \sumvd \alpha_v~x_v~\partial_{p_v}\ \in \MM~,\qquad
\label{e:cisimpliescizis}\\
\text{ if } \sumvd \alpha_v\ \partial_{p_v},~\sumvd
\beta_v\ \partial_{p_v}\in \MM~,
\quad &\text{ then }~~\sumvd \alpha_v~\beta_v\
\partial_{p_v}\ \in \MM~.\qquad\label{e:cisimpliesproducts}
\end{equs}
\end{itemize}
\end{lemma}
\begin{remark}
Observe that the assumption $d\ge1$ is crucial in the proof of (i).
Requiring the $\tilde{f}_v$ to be non-constant ensures that we can find
non-trivial double commutators, which is the crux of our analysis. See
\exref{ex:anharmon} for what goes wrong for harmonic potentials.
\end{remark}

\begin{proof}
(i).  By \eref{e:newnotation}
 and using \eref{e:implicationderivee} recursively $d-1$ times,
we find that
\begin{equ}
Y\equiv \sumv\bigl (\partial ^{\,d-1}
\tilde{f}_v \bigr)(x_v)\, \partial_{p_v} =\sumv \left(b_{v,d-1}
+  b_{vd}  x_v\right) \partial_{p_v}
\end{equ}
is in $\MM$.
But now, by \eref{e:scommvecteursdiagxu},
\begin{equ}
\scomm{Y}{Y/2} =   \sumv b_{vd}(b_{v,d-1}
+ b_{vd} x_v) \partial_{p_v} \in \MM~.
\end{equ}
Taking more double commutators with $Y/2$, we obtain for all $r\geq 1$:
\begin{equ}
 \sumv b_{vd}^{\,r}\,(b_{v,d-1}+ b_{vd} x_v) \partial_{p_v} \in \MM~.
\end{equ}
But the sum above is really only over  $\VVoldwd$ since
$b_{vd}\neq0$ only if $v\in\VVoldwd$. Moreover,
for these $v$, we have $b_{v,d-1} = 0$ since the polynomials are adjusted, so
that for all $i\geq 2$,
\begin{equ}\label{e:sumvdtildeau1plusi}
 \sumvd b_{vd}^{i} x_v \partial_{p_v} \in \MM~.
\end{equ}

Let $b\in \BB_d^{\,d}$. Using \lref{l:vandermonde} with $s=1$
 and with the set of distinct and non-zero values
 $\{b_{vd}~|~v\in \VVoldwd\}=\BB_d^{\,d}$ we find
real numbers $r_1, r_2,\dots,r_n$ (with  $n=\left|\BB_d^{\,d}\right|$)
such that $\sum_{i=1}^{n} r_i~b_{vd}^{i+1}$
 equals 1 if $b_{vd} = b$ and 0 when $b_{vd} \ne b$.
Thus,
\begin{equ}
\sum_{i=1}^{n} r_i \sumvd b_{vd}^{i+1} x_v \partial_{p_v} =
     \sum_{v\in \VVoldwd~:~ b_{vd}=b}\kern -1 em x_v\partial_{p_v}
\end{equ}
is in $\MM$ by \eref{e:sumvdtildeau1plusi}.
This together with \eref{e:implicationderivee} establishes
the second inclusion of \eref{e:spansuraxeta}, so that we have shown (i).

The statement (ii) follows by summing (i) over all $b\in\BB_d^{\,d}$.

Proof of (iii). Let us assume that $d\geq 2$ and that $\sumvd
\,\alpha_v\ \partial_{p_v}\in \MM$.
By \eref{e:spansuraxeta}, we have for each $b\in\BB_d^{\,d}$ that
\begin{equ}
{1\over b} \sum_{v\in \VVoldwd~:~b_{vd}=b}x_v\partial_{p_v} =
 \sum_{v\in \VVoldwd~:~b_{vd}=b} {x_v\over b_{vd}}\partial_{p_v}
\in \MM~.
\end{equ}
Taking the double
 commutator with $\sumvd\alpha_v\ \partial_{p_v}$ and summing
 over all $b\in\BB_d^{\,d}$\, shows that
\begin{equ}
U\equiv \sumvd {\alpha_v\over b_{vd} }\partial_{p_v} \in \MM~.
\end{equ}
Since we assume here $d\geq 2$, we have
$Z\equiv \sumv \bigl(\partial^{\,d-2}\tilde{f}_v\bigr)(x_v)\partial_{p_v}\in
\MM$.
But then,
\begin{equ}
\scomm{U}{Z} =
  \sumv {\alpha_v\over b_{vd}
}\bigl(\partial^{\,d-1}\tilde{f}_v\bigr)(x_v)\partial_{p_v}=
 \sumvd {\alpha_v\over
b_{vd} }\left({b_{v,d-1}}+b_{vd}\ x_v \right)\partial_{p_v}
\end{equ}
is also in $\MM$. Recalling that $b_{v, d-1} = 0$ for all $v\in \VVoldwd $,
we obtain \eref{e:cisimpliescizis}. Finally \eref{e:cisimpliesproducts} follows
from \eref{e:cisimpliescizis} and the double commutator
\begin{equ}
\scomm{\sumvd \alpha_v~x_v~\partial_{p_v}}{
 \sumvd \beta_v~\partial_{p_v}} = \sumvd \alpha_v~\beta_v~\partial_{p_v}~.
\end{equ}
This completes the proof.
\qed\end{proof}

With these preparations, we can now prove \tref{t:mainprop}.

\begin{proof}[of \tref{t:mainprop}]

We distinguish the cases $d=1$ and $d\ge2$.\hfill\break
{\bf Case $d=1$}: This case is easy. Since all the $\tilde{f}_v$ have degree
1,
we have that $\tilde{f}_v(x_v) = b_{v1}\ x_v$ for all $v\in \TT$, with
$b_{v1}\neq 0$.
Consequently, the sets $\TT_i$ consist of those $v$ which
have the same $b_{v1}$ (see \rref{r:ftildediff}).
Thus, we have by \eref{e:spansuraxeta} that for each $\TT_i$:
 \begin{equ}\label{e:d1}
\sum_{v\in\TT_i}\partial_{p_v}\in\MM~~
\text{  and }~~\sum_{v\in\TT_i}x_v\partial_{p_v}\in \MM \qquad \text{(if
$d=1$)}~.
 \end{equ}
This  shows that the conclusion of
\tref{t:mainprop} holds when $d=1$.\hfill\break
{\bf Case $d\ge2$}: In this case, \eref{e:spansuraxeta} is not enough.
 First, \eref{e:spansuraxeta} says nothing
 about the masses $v \in \TT\setminus \VVoldwd$, for which
$b_{vd}=0$. And second,
 \eref{e:spansuraxeta} provides us with no way to ``split"
the $\partial_{p_v}$ corresponding to a common (non-zero)
value $b$ of $b_{vd}$, even though the corresponding
$v$ might be inequivalent
due to some $b_{vk}$ with $k< d$. To fully make use of these coefficients,
we must develop some more advanced machinery.

\begin{definition}
We denote by $\PP_d$ the vector space of real polynomials
 in one variable of degree at most $d$. We consider
the operator $\G: \PP_d \to \PP_d$ defined by
\begin{equ}
 (\G v)(x)
\equiv (x\cdot v(x))'~,
\end{equ}
and we introduce the set of operators
\begin{equ}
\FF \equiv \mathrm{span}\{\G, \G^2, \dots, \G^{d+1}\}~.
\end{equ}
\end{definition}

Observe that by \eref{e:newnotation} and \eref{e:fullvector1dansMh} we have
\begin{equ}
\scomm{  \sum_{v\in\TT}  \tilde f_v(x_v) \partial _{p_v}}{ \sumvd
x_v\partial_{p_v} }
= \sumvd (\G \tilde f_v)(x_v)\partial_{p_v}  \in \MM~.
\end{equ}
Note that we obtain a sum over $\VVoldwd$ only. By taking more double
commutators
with $ \sumvd x_v\partial_{p_v}$, we find that $\sumvd (\G^k \tilde
f_v)(x_v)\partial_{p_v} $
is in $\MM$ for all $k\geq 1$. Thus, by the linear structure of $\MM$,
we obtain
\begin{lemma}\label{l:pok}
For all $\P\in\FF$, we have
   \begin{equ}
    \sumvd (\P\tilde{f}_v)(x_v)\partial_{p_v}  \in \MM~.
  \end{equ}
\end{lemma}
It is crucial to understand that it is the {\em same} operator $\P$ that  is
applied simultaneously to all the components,
 and that the components in
$\TT  \setminus \VVoldwd $ are ``projected out."

We now show that some very useful operators are in $\FF$.
\begin{proposition}\label{p:algebre}
The following hold:
\begin{itemize}
\item[(i)] The projector
\begin{equ}
\S_\ell : \PP_d \to \PP_d\,,\quad \sum_{i=0}^d b_i {x^i}/i! \mapsto  b_\ell
{x^\ell}/\ell!
\end{equ}
belongs to $\FF$ for all $\ell=0, \dots, d$.
\item[(ii)] The identity operator $\1$ acting on $\PP_d$ is in $\FF$.
\end{itemize}
\end{proposition}
\begin{proof} Consider the basis $B= (e_0, e_1, \dots, e_d)$ of $\PP_d$ where
$e_j(x) = {x^j}/j!$~.
Observe that for all $j \geq 0$ we have $\G  e_j = (j+1)e_j$, so that $\G$ is
diagonal in the basis $B$. Thus, $\G^k$  is represented by the matrix
 $\mathrm{diag}(1^k, 2^k, \dots, (d+1)^k)$ for all $k\geq 1$. Consequently, for
each
$\ell \in \{0, \dots, d\}$, there is by \lref{l:vandermonde} with $s=0$
a linear combination of $\G, \G^2, \dots, \G^{d+1}$ that is equal to $\S_\ell$.
This proves (i). Moreover, we have that $\sum_{\ell = 0}^d \S_\ell = \1$, so
that
the proof of (ii) is complete.
\qed\end{proof}

\begin{lemma}\label{l:lesajls} For all $\ell = 0, \dots, d$, and for
  each $b\in \BB_\ell^{\,d} = \{b_{v\ell} ~|~v\in \VVoldwd \}$ we have
  \begin{equ}\label{e:spansuratous}
    \sum_{v\in \VVoldwd ~:~b_{v\ell}=b}\partial_{p_v}  \in \MM~.
  \end{equ}
\end{lemma}
\begin{proof}

Let $\ell\in \{0, 1, \dots, d\}$.
Using \lref{l:pok} and \pref{p:algebre}(i) we find that $ \sumvd (b_{v\ell}
{x^\ell}/\ell!)\partial_{p_v}$
 is in $\MM$. Using \eref{e:implicationderivee} repeatedly,
we find that $\sumvd b_{v\ell} \partial_{p_v}$ is in $\MM$. Thus, by
 \eref{e:cisimpliesproducts},
\begin{equ}
\sumvd b_{v\ell}^i \partial_{p_v}\in \MM \qquad  \text{ for all $i\geq 1$}.
\end{equ}
Then, applying
\lref{l:vandermonde} to the set
 $\BB_\ell^{\,d}\setminus \{0\}$
and with $s = 0$,
we conclude that
  \begin{equ}\label{e:sumaksans0}
    \sum_{v\in \VVoldwd ~:~b_{v\ell}=b}\partial_{p_v}  \in \MM
\quad
\text{for all }\quad b\in \BB_\ell^{\,d}\setminus \{0\}~.
  \end{equ}
If $0\notin \BB_\ell^{\,d}$, we are done. Else, we obtain
that  \eref{e:spansuratous}
holds also for $b=0$ by summing the vector field \eref{e:sumaksans0}
over all $b\in\BB_\ell^{\,d}\setminus \{0\}$
and subtracting the result from $\sumvd \partial_{p_v}$ (which is in $\MM$ by
\eref{e:fullvector1dansMh}). This completes the proof.
\qed\end{proof}

Remember that by \rref{r:ftildediff}, a given equivalence class $\TT_i$
is either a subset of $\DD^{\,d}$ or completely disjoint from it.

\begin{lemma}\label{l:vd} Let $\TT_i$ be an equivalence class such that $\TT_i
\subset \DD^{\,d}$.
Then
\begin{equ} \label{e:xdpudpudansTi}
\sum_{v\in\TT_{i}}\partial_{p_v} \in \MM~, \quad \text{ and
}~~\sum_{v \in \TT_{i}}x_v\,\partial_{p_v}\in \MM~.
\end{equ}
\end{lemma}
\begin{proof}
All the polynomials $\tilde{f}_v$, $v\in \TT_i$ are equal. Thus, there
 are coefficients $c_\ell \in \BB_\ell^{\,d}$,  $\ell=0, 1, \dots, d$ such that
\begin{equ}\label{e:bfc}
\TT_{i}=  \bigcap_{\ell=0}^d \{ v\in\DD^{\,d}~|~ b_{v\ell} = c_\ell\}~.
\end{equ}
By \lref{l:lesajls}, we have for all $\ell=0, \dots, d$ that
  \begin{equ}\label{e:chaqueltinterdedans}
  \sum_{v\in \VVoldwd ~:~b_{v\ell}=c_\ell}\partial_{p_v}  \in \MM~.
  \end{equ}
Now observe that whenever two sets $B, B' \subset \VVoldwd$
are such that $\sum_{v\in B}\partial_{p_v} \in \MM$
and $\sum_{v\in B'}\partial_{p_v}\in\MM$, we
have by \eref{e:cisimpliesproducts} that
 $\sum_{v\in B\cap B'}\partial_{p_v} \in \MM$. Applying this
 recursively to the intersection in \eref{e:bfc}
and using \eref{e:chaqueltinterdedans} shows that
 $\sum_{v\in\TT_{i}}\partial_{p_v} \in \MM$. Using now
\eref{e:cisimpliescizis} implies that $\sum_{v\in\TT_{i}}x_v \partial_{p_v} \in
\MM$,
which completes the proof.
\qed\end{proof}

With these tools, we are now ready to complete the
proof of \tref{t:mainprop} (for the case $d\ge2$).
By \lref{l:vd}, we are done if $\VVoldwd = \TT$ (\ie, if all
the $\tilde{f}_v$, $v\in \TT$ have degree $d$).
If this is not the case, we proceed as follows.

Observe that \lref{l:pok} and \pref{p:algebre}(ii)
imply that $\sumvd \tilde{f}_v(x_v)\,\partial_{p_v}$ is in $\MM$.
Subtracting this from  \eref{e:newnotation} shows that
\begin{equ}
\sum_{v\in \TT\setminus \VVoldwd}  \tilde{f}_v(x_v)\partial_{p_v}\in \MM~.
\end{equ}
Thus, we can start the above procedure again with this new ``smaller" vector
field, each component being a polynomial of degree at most
\begin{equ}
d' \equiv \max_{v \in  \TT\setminus \VVoldwd} d_v~,
\end{equ}
with obviously $d'<d$. Defining then $\DD^{\,d'} = \{v\in \TT~|~d_v = d'\}$, we
get as in
\lref{l:vd} that \eref{e:xdpudpudansTi} holds for all $\TT_i \subset
\DD^{\,d'}$.
We then proceed like this inductively, dealing at each step with the components
 of highest degree and ``removing'' them, until all the remaining
 components
 have the same degree $d^-$ (which is equal to $\min_{v\in\TT}d_v$).
If $d^{-} \geq 2$ we obtain again as in \lref{l:vd}  that
 \eref{e:xdpudpudansTi} holds for all $\TT_i \subset \DD^{\,d^{-}}$.
And if $d^{-} = 1$, the conclusion follows from \eref{e:d1}.
Thus, \eref{e:xdpudpudansTi} holds for every equivalence class $\TT_i$,
regardless of the degree of the polynomials involved.
The proof of \tref{t:mainprop} is complete.
\qed\end{proof}

\begin{remark}\label{r:414}
Our method also covers the case where
each particle $v\in \VV$ can have an arbitrary positive mass $m_v$.
The proofs work the same way, if we replace the functions
$\tilde{f}_v$ with $\tilde{f}_v = V_{cv}''(x_v) / (m_c m_v)$.
Thus, if for example all the $V''_{cv}$, $v\in \TT$ are the same, but all the
particles
in $\TT$ have distinct masses, then all the new $\tilde{f}_v$ are different,
and
the particles in $\TT$ belong each to a separate $\TT_i$.
\end{remark}

\section{Controlling a network}\label{s:controlnet}

We now show how  \tref{t:mainprop} can be used recursively
to control a large class of networks. The idea is very simple: at each step
of the recursion, we apply \tref{t:mainprop} to a controllable particle (or a
set of such)
 in order to show that some neighboring vertices are also controllable.
Starting this procedure with the  particles
in $\VV_*$ (which are controllable by the definition of $\MM$),
we obtain under certain conditions that the whole network is
controllable.

In order to make the distinction clear, we will say that a particle $c$
is a {\em controller} if it is controllable and if we intend to use it as a
starting point to control other particles.

\begin{definition} Let $\JJ$ be the collection of jointly controllable sets
(\ie,
of sets $A\subset \VV$ such that $\sum_{v\in A}\partial_{p_v} \in \MM$,
and therefore also
$\sum_{v\in A}\partial_{q_v}$ by \lref{l:piimpliesqi}).
\end{definition}

Obviously, a particle $v$ is controllable iff $\{v\} \in \JJ$. The next lemma
shows what we ``gain" in $\JJ$ when we apply \tref{t:mainprop} to
 a controller $c$. Remember that the set $\TT^c$ of first neighbors of $c$ is
partitioned into equivalence classes $\TT^c_i$, as discussed in
\sref{s:neighcontroller}.

\begin{lemma}\label{l:lemmeintersections} Let $c \in \VV$ be a controller.
Then,
\begin{itemize}
\item[(i)]for all $i$,
\begin{equ}
\TT_i^c \in \JJ ~,
\end{equ}
\item[(ii)]for all $i$ and for all $A\in \JJ$ the sets
\begin{equ}\label{e:newsetsinJ}
A\cap \TT_i^c, \qquad A\setminus \TT_i^c \quad \text{ and } \quad \TT_i^c
\setminus A ~
\end{equ}
are in $\JJ$.
\end{itemize}
\end{lemma}
We illustrate some possibilities in \fref{f:refinement}.
\begin{figure}[ht]
\centering
\subfloat[]{\scalebox{0.7}{\begin{tikzpicture}[segment amplitude=10pt, line
     width=0.7pt]
\tikzstyle{every node}=[font=\large]

\draw[] (0,0) ellipse (15pt and 15pt);
\draw[] (-1.7,-1.2) ellipse (15pt and 15pt);
\draw[rotate around={-70:(1.25,1.35)}] (1.25,1.35) ellipse (30pt and 15pt);
\draw[rotate around={-65:(-2.15,0.6)}] (-2.15,0.6) ellipse (25pt and 15pt);
\draw[] (-0.5,2.5) ellipse (25pt and 30pt);

\path[draw=black, fill=white] (-1,2) circle (1mm);
\path[draw=black, fill=white] (0,2) circle (1mm);
\path[draw=black, fill=white] (1,2) circle (1mm);

\path[draw=black, fill=white] (-0.9,2.9) circle (1mm);
\path[draw=black, fill=white] (-0.2,3.1) circle (1mm);

\path[draw=black, fill=white] (1.5,0.7) circle (1mm);

\path[draw=black, fill=white] (-2.3,0.9) circle (1mm);
\path[draw=black, fill=white] (-2,0.3) circle (1mm);

\path[draw=black, fill=white] (0.9,-1.8) circle (1mm);
\path[draw=black, fill=white] (1.8,-0.9) circle (1mm);

\path[draw=black, fill=white] (-1.7,-1.2) circle (1mm);

\path[draw=black, fill=gray] (0,0) circle (1mm)node[right=4pt, above=1pt]{${c}$};
\end{tikzpicture}}
}
\hfill%
\subfloat[]{\scalebox{0.7}{\begin{tikzpicture}[segment amplitude=10pt, line
     width=0.7pt] 
\tikzstyle{every node}=[font=\large]

\draw[color=gray] (0,0) ellipse (15pt and 15pt);
\draw[color=gray](-1.7,-1.2) ellipse (15pt and 15pt);
\draw[color=gray, rotate around={-70:(1.25,1.35)}] (1.25,1.35) ellipse (30pt and 15pt);
\draw[color=gray, rotate around={-65:(-2.15,0.6)}] (-2.15,0.6) ellipse (25pt and 15pt);
\draw[color=gray] (-0.5,2.5) ellipse (25pt and 30pt);

\draw[snake=coil, segment amplitude=0.5pt] (0,0)   -- (-1,2);
\draw[snake=coil, segment amplitude=0.5pt] (0,0)   -- (0,2);
\draw[snake=coil, segment amplitude=0.5pt] (0,0)   -- (1,2);

\draw[snake=coil, segment amplitude=2pt] (0,0)   -- (0.9,-1.8);
\draw[snake=coil, segment amplitude=2pt] (0,0)   -- (1.8,-0.9);

\draw[snake=coil, segment amplitude=5pt] (0,0)   -- (-1.7,-1.2);

\node[rotate=-0,minimum height=20pt, minimum width=80pt,draw=black] (rect) at (0,2) {};
\node[rotate=-0,minimum height=20pt, minimum width=20pt,draw=black] (rect) at (-1.7,-1.2) {};
\node[rotate=45,minimum height=20pt, minimum width=60pt,draw=black] (rect) at (1.35,-1.35) {};

\path[draw=black, fill=white] (-1,2) circle (1mm);
\path[draw=black, fill=white] (0,2) circle (1mm);
\path[draw=black, fill=white] (1,2) circle (1mm);

\path[draw=black, fill=white] (-0.9,2.9) circle (1mm);
\path[draw=black, fill=white] (-0.2,3.1) circle (1mm);

\path[draw=black, fill=white] (1.5,0.7) circle (1mm);

\path[draw=black, fill=white] (-2.3,0.9) circle (1mm);
\path[draw=black, fill=white] (-2,0.3) circle (1mm);

\path[draw=black, fill=white] (0.9,-1.8) circle (1mm);
\path[draw=black, fill=white] (1.8,-0.9) circle (1mm);

\path[draw=black, fill=white] (-1.7,-1.2) circle (1mm);

\path[draw=black, fill=gray] (0,0) circle (1mm)node[right=10pt, above=0pt]{${c}$};
\end{tikzpicture}}
}
\hfill%
\subfloat[]{\scalebox{0.7}{\begin{tikzpicture}[segment amplitude=10pt, line
     width=0.7pt] 
\tikzstyle{every node}=[font=\large]

\draw[] (0,0) ellipse (15pt and 15pt);
\draw[] (-1.7,-1.2) ellipse (15pt and 15pt);
\draw[rotate around={-65:(-2.15,0.6)}] (-2.15,0.6) ellipse (25pt and 15pt);
\draw[rotate around={48:(1.4,-1.38)}] (1.4,-1.38) ellipse (30pt and 15pt);
\draw[] (-0.5,2) ellipse (25pt and 10pt);
\draw[rotate around={15:(-0.5,3)}] (-0.5,3) ellipse (25pt and 10pt);
\draw[] (1,2) ellipse (15pt and 15pt);
\draw[] (1.5,0.7) ellipse (15pt and 15pt);

\path[draw=black, fill=white] (-1,2) circle (1mm);
\path[draw=black, fill=white] (0,2) circle (1mm);
\path[draw=black, fill=white] (1,2) circle (1mm)node[right=5pt, above=2pt]{${c'}$};

\path[draw=black, fill=white] (-0.9,2.9) circle (1mm);
\path[draw=black, fill=white] (-0.2,3.1) circle (1mm);

\path[draw=black, fill=white] (1.5,0.7) circle (1mm)node[right=5pt, above=2pt]{${c''}$};

\path[draw=black, fill=white] (-2.3,0.9) circle (1mm);
\path[draw=black, fill=white] (-2,0.3) circle (1mm);

\path[draw=black, fill=white] (0.9,-1.8) circle (1mm);
\path[draw=black, fill=white] (1.8,-0.9) circle (1mm);

\path[draw=black, fill=white] (-1.7,-1.2) circle (1mm);

\path[draw=black, fill=gray] (0,0) circle (1mm)node[right=4pt, above=1pt]{${c}$};
\end{tikzpicture}}
}
\caption{a: A controller $c$ and a few sets in $\JJ$, shown as ovals.
b: The equivalence classes $\TT^c_i$ are shown as rectangles.
(Only the edges incident to $c$ are shown.)
c: New sets ``appear" in $\JJ$. In particular,  $c'$ and $c''$
 are controllable.}
\label{f:refinement}
\end{figure}
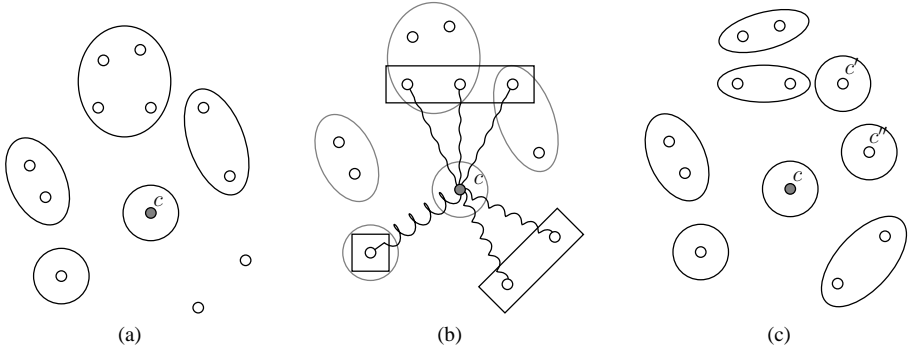

\begin{proof}
(i). This is an immediate consequence of \eref{e:vectorfield} and the
definition
of $\JJ$.

\noindent(ii). We consider an equivalence class $\TT^c_i$ and a set $A\in \JJ$.
By \eref{e:vectorfieldx} we find that
\begin{equ}\null
[\sum_{v\in\TT_i^c} (q_c-q_v+\delta _{cv})\partial _{p_v}, \sum_{v\in A}
\partial_{q_v}]
= \sum_{v\in A\cap \TT^c_i}\partial_{p_v} - \1_{c\in A}\cdot  \sum_{v\in
\TT^c_i}\partial_{p_v}
\end{equ}
is in $\MMà$. By the linear structure of $\MM$ and since
 $\sum_{v\in \TT^c_i}\partial_{p_v} $ is in $\MM$ by \eref{e:vectorfield},
we can discard the second term and we find $\sum_{v\in A\cap
\TT^c_i}\partial_{p_v}
\in \MM$. This proves that $A\cap \TT^c_i$ is in $\JJ$. Then,
subtracting $\sum_{v\in A\cap \TT^c_i}\partial_{p_v} $ from $\sum_{v\in A}
\partial_{p_v}$ (resp.~from
$\sum_{v\in \TT^c_i} \partial_{p_v}$) shows  that
 $\sum_{v\in A \setminus \TT^c_i} \partial_{p_v}$ (resp.~$\sum_{v\in  \TT^c_i
\setminus A}\partial_{p_v}$)
is in $\MM$, which completes the proof of (ii).
\qed\end{proof}

We can now give an algorithm that applies \lref{l:lemmeintersections}
recursively, and that can be used to show that a large variety
of networks is controllable.

\begin{proposition} Consider the following algorithm that builds
 step by step a collection of sets $\WW \subset \JJ$
and a list of controllable particles $K$.

Start with $\WW = \left\{\{v\}~|~v\in \VV_*\right\}$ and put
(in any order) the vertices of $\VV_*$ in $K$.
\begin{enumerate}
\item Take the first unused controller $c \in K$.
\item Add each equivalence class $\TT^c_i$ to $\WW$.
\item For each $\TT^c_i$ and  each $A\in \WW$
 add the sets of \eref{e:newsetsinJ} to $\WW$.
\item If  in 2. or 3. new singletons appear in $\WW$, add
the corresponding vertices (in any order) at the end of $K$.
\item Consider $c$ as used. If there is an unused controller in $K$, start
again at 1. Else, stop.
\end{enumerate}
We have the following result: if in the end $K$ contains all the vertices of
$\VV$, then
 the network is controllable.
\end{proposition}
\begin{proof}By \lref{l:lemmeintersections}, the collection $\WW$ remains at
each step
a subset of $\JJ$, and $K$ contains only controllers. Thus, the result holds
by construction.
\qed\end{proof}

The algorithm stops after at most $|\VV|$ iterations, and one can
 show that the result does not depend on the order in which we use the
controllers.
This algorithm is probably the easiest to implement, but does not
give much insight into what really makes
 a network controllable with our criteria. For this reason, we now formulate a
 similar result in terms of equivalence with respect to a {\em set} of
controllers,
 which underlines the role of the ``cooperation'' of several controllers.

\begin{definition}\label{d:covering}
We consider a set $C$ of controllers and denote by $\NN(C)$
the set of first neighbors of $C$ that are not themselves in $C$. We say that
two particles $v, w\in \NN(C)$ are \ctop if $v$ and $w$ are connected to $C$
in exactly the same way, \ie, if for every $c\in C$ the
 edges $\{c, v\}$ and $\{c, w\}$ are either both
present or both absent.

Moreover, we say that $v$ and $w$ are $C$-equivalent if they are
\ctop, and if in addition, for each $c\in C$ that is linked
to $v$ and $w$,
we have that $v$ and $w$ are equivalent with respect to $c$ (\ie, there is
a $\delta\in \real$ such that
$V''_{cv}(\argcdot) = V''_{cw}(\argcdot + \delta)$).
\end{definition}
The $C$-equivalence classes form a refinement of the
sets of \ctop, see \fref{f:refinementre}.

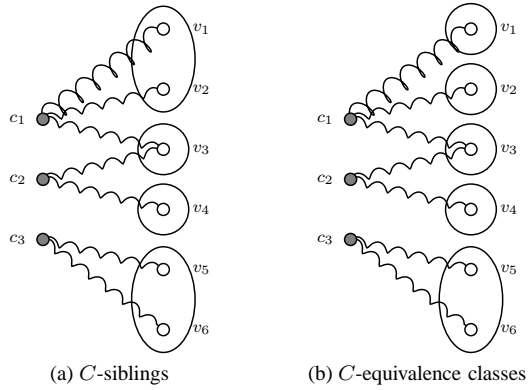
\begin{figure}[ht]
\centering

\subfloat[\ctop]{

\scalebox{0.8}{\begin{tikzpicture}[segment amplitude=10pt, line
     width=0.7pt] 


\draw[rotate around={0:(1.4,-1.38)}] (2,-0) ellipse (15pt and 25pt);
\draw[rotate around={0:(1.4,-1.38)}] (2,-2.5) ellipse (12pt and 12pt);
\draw[rotate around={0:(1.4,-1.38)}] (2,-1.5) ellipse (12pt and 12pt);
\draw[rotate around={0:(1.4,-1.38)}] (2,-4) ellipse (15pt and 25pt);

 \draw[snake=coil, segment amplitude=2pt] (0,-1)   -- (2,-0.5);
 \draw[snake=coil, segment amplitude=2pt] (0,-2)   -- (2,-1.5);
 \draw[snake=coil, segment amplitude=2pt] (0,-2)   -- (2,-2.5);
 \draw[snake=coil, segment amplitude=2pt] (0,-3)   -- (2,-3.5);

 \draw[snake=coil, segment amplitude=2pt] (0,-1)   -- (2,-1.5);
 \draw[snake=coil, segment amplitude=2pt] (0,-3)   -- (2,-4.5);

 \draw[snake=coil, segment amplitude=5pt] (0,-1)   -- (2,0.5);

\path[draw=black, fill=gray] (0,-1) circle (1mm)node[left=5pt]{${c_1}$};
\path[draw=black, fill=gray] (0,-2) circle (1mm)node[left=5pt]{${c_2}$};
\path[draw=black, fill=gray] (0,-3) circle (1mm)node[left=5pt]{${c_3}$};

\path[draw=black, fill=white] (2,0.5) circle (1mm)node[right=11pt]{$v_1$};
\path[draw=black, fill=white] (2,-0.5) circle (1mm)node[right=11pt]{$v_2$};
\path[draw=black, fill=white] (2,-1.5) circle (1mm)node[right=11pt]{$v_3$};
\path[draw=black, fill=white] (2,-2.5) circle (1mm)node[right=11pt]{$v_4$};
\path[draw=black, fill=white] (2,-3.5) circle (1mm)node[right=11pt]{$v_5$};
\path[draw=black, fill=white] (2,-4.5) circle (1mm)node[right=11pt]{$v_6$};

\end{tikzpicture}}

}
\qquad\qquad
\subfloat[$C$-equivalence classes]{
\scalebox{0.8}{\begin{tikzpicture}[segment amplitude=10pt, line
     width=0.7pt] 

\draw[rotate around={0:(1.4,-1.38)}] (2,-2.5) ellipse (12pt and 12pt);
\draw[rotate around={0:(1.4,-1.38)}] (2,-1.5) ellipse (12pt and 12pt);
\draw[rotate around={0:(1.4,-1.38)}] (2,0.5) ellipse (12pt and 12pt);
\draw[rotate around={0:(1.4,-1.38)}] (2,-0.5) ellipse (12pt and 12pt);
\draw[rotate around={0:(1.4,-1.38)}] (2,-4) ellipse (15pt and 25pt);

 \draw[snake=coil, segment amplitude=2pt] (0,-1)   -- (2,-0.5);
 \draw[snake=coil, segment amplitude=2pt] (0,-2)   -- (2,-1.5);
 \draw[snake=coil, segment amplitude=2pt] (0,-2)   -- (2,-2.5);
 \draw[snake=coil, segment amplitude=2pt] (0,-3)   -- (2,-3.5);

 \draw[snake=coil, segment amplitude=2pt] (0,-1)   -- (2,-1.5);
 \draw[snake=coil, segment amplitude=2pt] (0,-3)   -- (2,-4.5);

 \draw[snake=coil, segment amplitude=5pt] (0,-1)   -- (2,0.5);

\path[draw=black, fill=gray] (0,-1) circle (1mm)node[left=5pt]{${c_1}$};
\path[draw=black, fill=gray] (0,-2) circle (1mm)node[left=5pt]{${c_2}$};
\path[draw=black, fill=gray] (0,-3) circle (1mm)node[left=5pt]{${c_3}$};

\path[draw=black, fill=white] (2,0.5) circle (1mm)node[right=11pt]{$v_1$};
\path[draw=black, fill=white] (2,-0.5) circle (1mm)node[right=11pt]{$v_2$};
\path[draw=black, fill=white] (2,-1.5) circle (1mm)node[right=11pt]{$v_3$};
\path[draw=black, fill=white] (2,-2.5) circle (1mm)node[right=11pt]{$v_4$};
\path[draw=black, fill=white] (2,-3.5) circle (1mm)node[right=11pt]{$v_5$};
\path[draw=black, fill=white] (2,-4.5) circle (1mm)node[right=11pt]{$v_6$};

\end{tikzpicture}}

}
\caption{
Illustration of \dref{d:covering}. We assume that all the springs are
identical,
except for the edge $\{c_1, v_1\}$. The particles $v_1,\dots, v_6$ form 4 sets
of \ctop,
with $C=\{c_1, c_2, c_3\}$.
The one containing $v_1$ and $v_2$ is further split
into two $C$-equivalence classes, since
$v_1$ and $v_2$ are by assumption inequivalent with respect to $c_1$.}
\label{f:refinementre}
\end{figure}

\begin{proposition}\label{p:controllersequiv} Let $C$ be a set of controllers.
 Then, for each $C$-equivalence class
 $U\subset \NN(C)$, we have $U\in \JJ$.
\end{proposition}
\begin{proof}See \fref{f:fstar}. Let $U = \{v_1, \dots, v_n\} \subset \NN(C)$
be a $C$-equivalence class.
We denote by $c_1, \dots, c_k$ the controllers in $C$ that are linked to $v_1$,
 and therefore also to $v_2, \dots, v_n$, since the elements of $U$ are
 \ctop. For each $j\in \{1, \dots, k\}$, there is a $\TT^{c_j}_{i}$
with $v_1, \dots, v_n \in \TT^{c_j}_{i}$, and we define $\SS_j \equiv
\TT^{c_j}_{i}\setminus C$.
We consider the set
\begin{equ}
\widehat{U}\equiv \bigcap_{j=1}^k \SS_j~.
\end{equ}
Clearly, $U\subset \widehat{U}$,
 and $\widehat{U} \in \JJ$ by \lref{l:lemmeintersections}. We have
$\widehat{U} = \{v_1, \dots, v_n, v^*_1, \dots, v^*_r\}$, where the $v^*_j$
are those particles that are equivalent to $v_1, \dots, v_n$ from the
point of view of $c_1, \dots, c_k$, but that are also connected to
some controller(s) in $C\setminus\{c_1, \dots, c_k\}$.
In particular, for each $j \in \{1, \dots, r\}$, there is a $c^*_j \in
C\setminus\{c_1, \dots, c_k\}$
and an $i$ such that $v^*_j$ is in $\SS^*_j \equiv  \TT^{c^*_j}_i$. By
construction,
$\SS_j^* \cap U=\emptyset$. Thus,
\begin{equ}
 \widehat{U} \setminus \bigcup_{j=1}^r \SS_j^* = U~.
\end{equ}
Starting from $\widehat{U} \in \JJ$ and  removing one by one the $\SS_j^*$,
we find by  \lref{l:lemmeintersections}(ii) that $U$ is indeed in
$\JJ$, as we claim.
\qed
\end{proof}
\begin{figure}[ht]
\centering

\scalebox{1.0}{\begin{tikzpicture}[segment amplitude=10pt, line
     width=0.7pt]


 \draw[snake=coil, segment amplitude=0pt] (0,-0.75)   -- (2,-0.75);
 \draw[snake=coil, segment amplitude=0pt] (0,-0.75)   -- (2,-1.5);
 \draw[snake=coil, segment amplitude=0pt] (0,-0.75)   -- (2,-2.5);
 \draw[snake=coil, segment amplitude=0pt] (0,-0.75)   -- (2,-3.5);
 
 \draw[snake=coil, segment amplitude=0pt] (0,-2)   -- (2,-1.5);
 \draw[snake=coil, segment amplitude=0pt] (0,-2)   -- (2,-2.5);
 \draw[snake=coil, segment amplitude=0pt] (0,-2)   -- (2,-3.5);

 \draw[snake=coil, segment amplitude=0pt] (0,-3.25)   -- (2,-3.5);
 \draw[snake=coil, segment amplitude=0pt] (0,-3.25)   -- (2,-0.75);

\path[draw=black, fill=gray] (0,-0.75) circle (1mm)node[left=5pt]{${c_1}$};
\path[draw=black, fill=gray] (0,-2) circle (1mm)node[left=5pt]{${c_2}$};
\path[draw=black, fill=gray] (0,-3.25) circle (1mm)node[left=5pt]{${c_1^*}$};

\path[draw=black, fill=white] (2,-0.75) circle (1mm)node[right=11pt]{$w$};
\path[draw=black, fill=white] (2,-1.5) circle (1mm)node[right=11pt]{$v_1$};
\path[draw=black, fill=white] (2,-2.5) circle (1mm)node[right=11pt]{$v_2$};
\path[draw=black, fill=white] (2,-3.5) circle (1mm)node[right=11pt]{$v_1^*$};

\end{tikzpicture}}
\caption{Illustration of the proof of \pref{p:controllersequiv} for identical
 springs. We consider the $C$-equivalence class $U=\{v_1, v_2\}$,
where $C$ contains $c_1, c_2, c_1^*$ and possibly other particles (not shown)
that are not linked to $v_1, v_2$.
With the notation of the proof, we have $\SS_1 = \{w, v_1, v_2, v^*_1\}$ and
$\SS_2 = \{v_1, v_2, v_1^*\}$
so that  $\widehat
 U=\SS_1\cap \SS_2 = \{v_1,v_2,v_1^*\}$. Since $v_1^*$ belongs to $\SS^*_1 =
\{w, v_1^*\}$,
we find $\widehat{U} \setminus \SS^*_1 = U$.
}\label{f:fstar}
\end{figure}

An immediate consequence is
\begin{corollary}\label{c:aloneinequiv} Let $C$ be a set of controllers.
If a vertex $v\in \NN(C)$
is alone in its $C$-equivalence class, then it is controllable.
\end{corollary}

Applying this recursively, we obtain

\begin{theorem}\label{t:resultCk}
We start with $C_0 \equiv \VV_*$. For each $k\geq 0$, let
\begin{equ}
C_{k+1} \equiv C_k\cup \{v\in \NN(C_k)~|~\text{ no other vertex
in $\NN(C_k)$ is $C_k$-equivalent to $v$}\}~.
\end{equ}
Then, if $C_k = \VV$ for some $k\geq 0$, the network is controllable.
\end{theorem}

\begin{proof}
By \Cref{c:aloneinequiv} we have that each $C_k$ contains only
controllers (remember also that $\VV_*$ contains only controllers
by the definition of $\MM$). Thus if $C_k = \VV$ for some $k\geq 0$ we find
that
all vertices are controllers, which is what we claim.
\qed\end{proof}

\section{Examples}\label{s:examples}

In this section we illustrate by several examples the range
of our controllability criteria.
\begin{example}\label{ex:61}
A single controller $c$ can control \emph{several} particles if the interaction
potentials
between $c$ and its neighbors have pairwise inequivalent second derivative. See
\fref{f:61}.
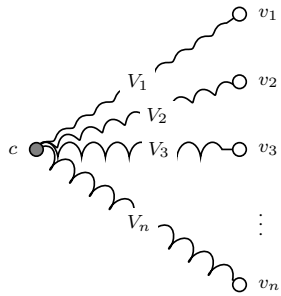
\begin{figure}[ht]
\centering
\scalebox{0.9}{\begin{tikzpicture}[segment amplitude=10pt, line
     width=0.7pt] 

 \draw[snake=coil, segment amplitude=4pt] (1,1)   --node[midway, right=0.5pt, below=-3.8pt, fill=white]{${V_n}$}  (4,-1);
 \draw[snake=coil, segment amplitude=3pt] (1,1)   --node[midway, right=1.3pt, fill=white]{${V_3}$}   (4,1);
 \draw[snake=coil, segment amplitude=2pt] (1,1)   --node[midway, right=0.5pt, fill=white]{${V_2}$}   (4,2);
 \draw[snake=coil, segment amplitude=1pt] (1,1)   -- node[midway, fill=white]{${V_1}$}  (4,3);

\path[draw=black, fill=gray] (1,1) circle (1mm)node[left=5pt]{${c}$};
\path[draw=black, fill=white] (4,3) circle (1mm)node[right=5pt]{${v_1}$};
\path[draw=black, fill=white] (4,2) circle (1mm)node[right=5pt]{${v_2}$};
\path[draw=black, fill=white] (4,1) circle (1mm)node[right=5pt]{${v_{3}}$};
\path[draw=black, fill=white] (4,-1) circle (1mm)node[right=5pt]{${v_n}$};
\path[draw=black, fill=white] (4,0) node[right=5pt]{${\vdots}$};

\end{tikzpicture}}
\caption{If no two springs are
equivalent, the $v_i$ are controllable.
Springs from the $v_i$ to other particles or from one $v_i$ to
another may exist but  are not shown. They do not change the conclusion.}
\end{figure}\label{f:61}
\end{example}

The example above does not use the topology of the network (\ie, the notion
of siblings), but only the inequivalence due
to the second derivative of the potentials. We have the following immediate
generalization, which we formulate as

\begin{theorem}\label{t:singletons}
Assume that  $\GG$ is connected, that $\VV_*$ is not empty, and that for each
$v\in \VV$, the first neighbors of $v$ are all pairwise
inequivalent with respect to $v$ (\ie, no two distinct neighbors $u, w$ of $v$
are such that $V''_{vu}(\argcdot) = V''_{vw}(\argcdot+\delta)$ for some
constant $\delta\in \real$). Then, the network is controllable.
\end{theorem}
\begin{proof} We use \tref{t:resultCk}. Observe that under these
assumptions, we have at each step $C_{k+1} = C_k \cup \NN(C_k)$.
Thus, since the network is connected, there is indeed a $k\geq 0$
such that $C_k = \VV$.
\qed\end{proof}

One can restate \tref{t:singletons} as a genericity condition:
\begin{corollary}\label{c:genericity}
Assume that  $\GG$ is connected, that $\VV_*$ is not empty, and that
for each $e\in \EE$ the degree of the polynomial $V_e$
is fixed (and is at least 3). Then, $\GG$ is almost surely controllable
if we pick the coefficients of each $V_e$ at random according
 to a probability law that is
absolutely continuous w.r.t.~Lebesgue.
\end{corollary}

\begin{example} The 1D chain (shown in \fref{f:chaine}) is
  controllable. Our theory applies when the interactions are polynomials of
degree at
  least 3; for a somewhat different variant, see \cite{EPR1999}.
To apply our criteria,
we start with $C = \{c\}$. Clearly, $v_1$ is alone in its
$C$-equivalence class, and
is therefore controllable by \Cref{c:aloneinequiv}. We then take $C'=\{c,
v_1\}$. Since
$v_2$ is alone in its $C'$-equivalence class, it is also controllable.
 Continuing like this, we find that
the whole chain is controllable.
\begin{figure}[ht]
\centering
\scalebox{1}{\begin{tikzpicture}[segment amplitude=10pt, line width=0.7pt] 

 \draw[snake=coil, segment amplitude=1pt] (0,0)   --   (2,0);
 \draw[snake=coil, segment amplitude=1pt] (2,0)   --   (4,0);
 \draw[snake=coil, segment amplitude=1pt] (4,0)   --   (6,0);
 \draw[snake=coil, segment amplitude=1pt] (6,0)   --   (8,0);
 \draw[snake=coil, segment amplitude=1pt] (8,0)   --   (10,0);

\path[draw=black, fill=gray] (0,0) circle (1mm)node[left=5pt]{${c}$};
\path[draw=black, fill=white] (2,0) circle (1mm)node[below=5pt]{${v_1}$};
\path[draw=black, fill=white] (4,0) circle (1mm)node[below=5pt]{${v_2}$};
\path[draw=black, fill=white] (6,0) circle (1mm)node[below=5pt]{${v_3}$};
\path[draw=black, fill=white] (8,0) circle (1mm)node[below=5pt]{${v_4}$};
\path[draw=black, fill=white] (10,0) circle (1mm)node[below=5pt]{${v_5}$};
\end{tikzpicture}}
\caption{A one-dimensional chain.}\label{f:chaine}
\end{figure}
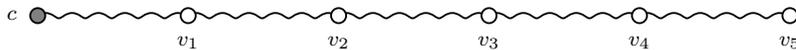
\end{example}

Observe that the chain described in the example above is controllable whether
some pairs of springs are equivalent or not. There are in fact many networks
that are
controllable thanks to their topology alone, regardless of the potentials. In particular, we have

\begin{example}[Physically relevant networks]We consider the network in \fref{f:carre}(a) and we start with
$C=\{c_1, \dots , c_4\}$ (\ie, we assume that the vertices in the first column are controllers).
Since no two vertices in the second column are \ctop, 
they each belong to a distinct $C$-equivalence class, and therefore by  \Cref{c:aloneinequiv}
 they are controllable (regardless of the potentials).
Let us now denote by $C'$ the set of all vertices
in the first two columns, which are controllable as we have just seen. Repeating 
the argument above, we obtain that the vertices in the third column are controllable.
Continuing like this, we gain control of the whole network. In the same way, one also
 easily obtains that the networks in
\fref{f:carre}(b-d) are controllable thanks to their topology alone.
\begin{figure}[ht]
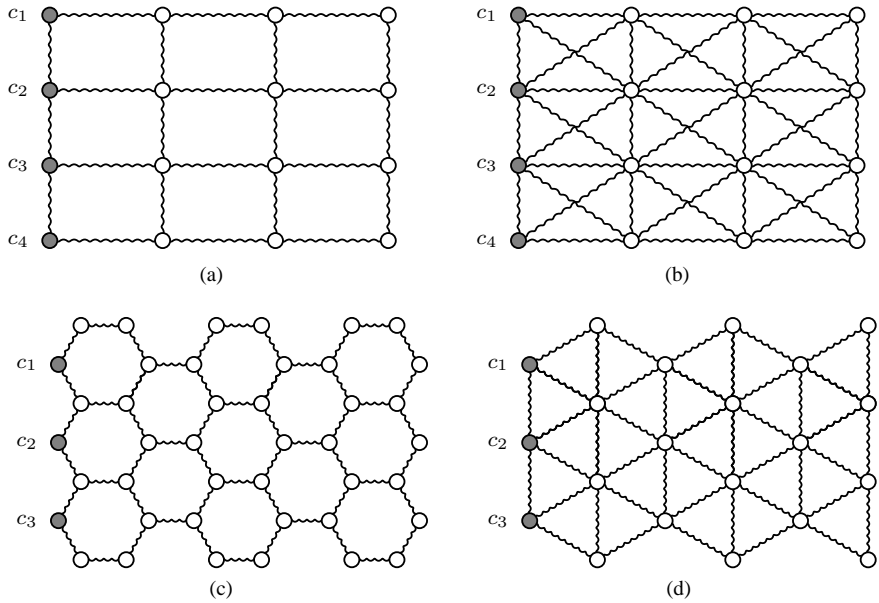

\input fig3
\caption{Four networks that are controllable by their topology alone,
regardless of the potentials (as long as they are polynomials of degree at
least 3).}\label{f:carre}
\end{figure}
\end{example}

\section{Limitations and extensions}

Our theory is local in the sense that the central tool
(\tref{t:mainprop}) involves only a controller and its first neighbors.
When we ``walk through the graph," starting from $\VV_*$ and taking
at each step control of more particles, we only look at the interaction
potentials
that involve the particles we already control and their first neighbors.
 We never look ``farther" in the graph.
This makes our criteria quite easy to apply, but this is also the
main limitation of our theory, as illustrated in

\begin{example}\label{ex:limites1}We consider the network shown in
\fref{f:exempleslimites}, where $c$ is a controller.
 If $V''_{cv_1}$ and $V''_{cv_3}$ are equivalent, then our theory fails to say
anything about the controllability of the network. In order to draw any
conclusion,
 one has to look at ``what comes next" in the network. Of course,
if the lower branch is an exact copy of the upper one
(\ie, if the interaction and pinning potentials
are the same), then the network is truly uncontrollable, and this is obvious
for symmetry reasons. However, without such an ``unfortunate"
symmetry, the network may still be controllable. Indeed, by
the study above, we know that the vector field $Y\equiv \partial_{q_{v_1}} +
\partial_{q_{v_3}}$
is in $\MM$. By commuting with ${X}_0$ and subtracting some contributions
already in $\MM$, one easily obtains that the vector field
\begin{equ}
U''_{v_1} \partial_{p_{v_1}} + U''_{v_3}\partial_{p_{v_3}} + V''_{v_1v_2}\cdot
(\partial_{p_{v_1}}-\partial_{p_{v_2}})+
V''_{v_3v_4}\cdot(\partial_{p_{v_3}}-\partial_{p_{v_4}})
\end{equ}
is in $\MM$.
 Observe that now the pinning potentials $U_{v_1}$ and $U_{v_3}$
as well as the interaction potentials $V_{v_1v_2}$ and $V_{v_3v_4}$
come into play.
Taking first commutators with $Y$ and then taking double commutators
among the obtained vector fields, one obtains further vector fields of the form
$\sum_{i=1}^4 g_i(q_{v_1}, q_{v_2}, q_{v_3}, q_{v_4})\partial_{p_{v_i}}$, where
the $g_i$ involve
derivatives and products of the potentials mentioned above. In many cases,
these
are enough to prove that the network in \fref{f:exempleslimites} is
controllable, even
though our theory fails to say so.
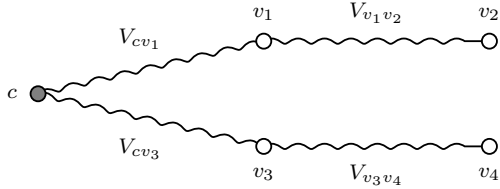
\begin{figure}[ht]
\centering
\scalebox{1}{\begin{tikzpicture}[segment amplitude=10pt, line
     width=0.7pt] 

 \draw[snake=coil, segment amplitude=1pt] (1,0)   --node[left=4pt, above=4pt, fill=white]{$V_{cv_1}$}   (4,0.7);
 \draw[snake=coil, segment amplitude=1pt] (1,0)   --node[left=4pt, below=4pt, fill=white]{$V_{cv_3}$}   (4,-0.7);
 \draw[snake=coil, segment amplitude=1pt] (4,-0.7)   -- node[below=4pt, fill=white]{$V_{v_3v_4}$}  (7,-0.7);
 \draw[snake=coil, segment amplitude=1pt] (4,+0.7)   -- node[above=4pt, fill=white]{$V_{v_1v_2}$}  (7,+0.7);

\path[draw=black, fill=gray] (1,0) circle (1mm)node[left=5pt]{${c}$};
\path[draw=black, fill=white] (4,0.7) circle (1mm)node[above=5pt]{${v_1}$};
\path[draw=black, fill=white] (4,-0.7) circle (1mm)node[below=5pt]{${v_3}$};
\path[draw=black, fill=white] (7,-0.7) circle (1mm)node[below=5pt]{${v_4}$};
\path[draw=black, fill=white] (7,0.7) circle (1mm)node[above=5pt]{${v_2}$};
\end{tikzpicture}}
\caption{The network used in \exref{ex:limites1}. If $V_{cv_3}$ is
  equivalent to $V_{cv_1}$ our theory does not allow to conclude, but
  the network might still be controllable.}\label{f:exempleslimites}
\end{figure}
\end{example}

One question that might arise is: why does only the  {\em second}
derivative of the interaction potentials enter the theory? The next example
shows that this issue is related to the notion of locality mentioned above.

\begin{example}\label{ex:limitesderivee}We consider the network
in
\fref{f:firstder}, where $c$ is a controller. We study the case
where
\begin{equs}[4]
V_{cv}(q_c-q_v) &= (q_c-q_v)^4, \quad&
 \quad U_{v}(q_v) &= q_v^6~,\\
V_{cw}(q_c-q_w) &= (q_c-q_w)^4 + a\cdot (q_c-q_w) , \quad &
 \quad U_{w}(q_w) &= q_w^6 +b\cdot q_w~,
\end{equs}
for some constants $a$ and $b$. The terms in $a$ and $b$ act as constant forces
on $c$ and $w$. Since $V''_{cv}\sim V''_{cw}$, the particles
$v$ and $w$ are equivalent with respect to $c$ by our definition. Thus, our
theory
fails to say anything. We seem to be missing
the fact that when $a\neq 0$, the particles $v$ and $w$ can be told apart due
to the first
derivative of the potentials. However, having $a\neq 0$
is not enough; the controllability of the network also depends on $b$.
Indeed, if $a=b$, the vector field $X_0$ is symmetric in $v$ and $w$, and
therefore
the network is genuinely uncontrollable. If now $a\neq b$,
 we have checked, by following a
different strategy of taking commutators, that the network
is controllable. Consequently, when two potentials have equivalent second
derivative, but inequivalent first derivative, no conclusion can be drawn
in general without knowing more about the network (here, it is one of the
pinning potentials,
but in more complex situations, it can be some subsequent springs).
\begin{figure}[ht]
\centering

\scalebox{1}{\begin{tikzpicture}[segment amplitude=10pt, line width=0.7pt]

 \draw[snake=coil, segment amplitude=1pt] (1,0)   --node[left=4pt, above=4pt, fill=white]{$V_{cv}$}   (4,0.7);

 \draw[snake=coil, segment amplitude=1pt] (1,0)   --node[left=4pt, below=4pt, fill=white]{$V_{cw}$}   (4,-0.7);

\path[draw=black, fill=gray] (1,0) circle (1mm)node[left=5pt]{${c}$};

\path[draw=black, fill=white] (4,0.7) circle (1mm)node[above=5pt]{${v}$};

\path[draw=black, fill=white] (4,-0.7) circle (1mm)node[below=5pt]{${w}$};

\end{tikzpicture}}
\caption{The network discussed in \exref{ex:limitesderivee}.}\label{f:firstder}
\end{figure}
\end{example}

Our theory applies only to strictly anharmonic systems, since we assume that
the interaction
potentials have degree at least 3. The next example shows what can go wrong if
we drop
this assumption. Again, this is related to the locality of our criteria.

\begin{example}\label{ex:anharmon}We consider the harmonic system shown in
\fref{f:contreexharmon}.
The vertex $c$ is a controller, and all the pinning potentials are equal and
harmonic,
 \ie, of the form $\lambda x^2/2$. The interaction potentials are also
harmonic. The spring $\{c, v_1\}$ has coupling constant
$2$, the springs $\{c, v_2\}$ and $\{v_2, v_3\}$ have coupling constant $1$
and the spring $\{v_3, v_4\}$
has coupling $k > 0$. Since $V_{cv_1}'' \equiv 2$ and $V_{cv_2}''\equiv 1$, the
particles
 $v_1$ and $v_2$ are inequivalent with respect to $c$.
Yet,  this is not enough to obtain that they are controllable
(unlike in the strictly anharmonic case covered by our theory). With standard
methods for
harmonic systems, it can
be shown that the network is controllable iff $k\neq 2$. When $k=2$, one of the
eigenmodes decouples from the controller $c$,
and no particle except $c$ is controllable. Thus, one cannot obtain that
$v_1$ and $v_2$ are controllable without knowing more about the potentials
that come farther in the graph.

\begin{figure}[ht]
\centering
\scalebox{1}{\begin{tikzpicture}[segment amplitude=10pt, line
     width=0.7pt] 

 \draw[snake=coil, segment amplitude=1pt] (1,0)   --node[left=3pt, above=4pt, fill=white]{${2}$}   (3,0.5);
 \draw[snake=coil, segment amplitude=1pt] (1,0)   --node[left=3pt, below=4pt, fill=white]{${1}$}   (3,-0.5);
 \draw[snake=coil, segment amplitude=1pt] (3,-0.5)   -- node[below=4pt, fill=white]{${1}$}  (5,-0.5);
 \draw[snake=coil, segment amplitude=1pt] (5,-0.5)   -- node[below=4pt, fill=white]{${k}$}  (7,-0.5);

\path[draw=black, fill=gray] (1,0) circle (1mm)node[left=5pt]{${c}$};
\path[draw=black, fill=white] (3,0.5) circle (1mm)node[above=5pt]{${v_1}$};
\path[draw=black, fill=white] (3,-0.5) circle (1mm)node[below=5pt]{${v_2}$};
\path[draw=black, fill=white] (5,-0.5) circle (1mm)node[below=5pt]{${v_3}$};
\path[draw=black, fill=white] (7,-0.5) circle (1mm)node[below=5pt]{${v_4}$};
\end{tikzpicture}}
\caption{A harmonic network that may or may not be controllable
depending on the value of the coupling constant $k$.}\label{f:contreexharmon}
\end{figure}
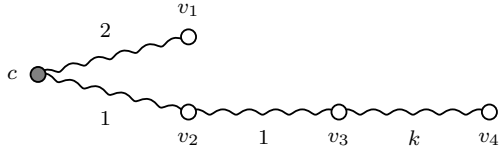
\end{example}

\begin{remark} As presented here, our method only works when 
the motion of each particle is 1D. To some extent, our results can be generalized
to higher dimensions. For example, one can check that in any dimension $r\geq 1$, 
the network of \fref{f:61} with potentials
 $V_k({\mathbf x}_k) = a_k~(x_{k,1}^2 + \dots + x_{k,r}^2 )^2$ , $k=1, \dots, n$,
is controllable when the $a_k$ are all distinct and non-zero. But for 
generic polynomial potentials, the situation is more complicated: 
taking multiple commutators does not always lead to tractable expressions
(in particular, we do not have the nice 
form \eref{e:scommvecteursdiagxu} for double commutators anymore).
Further research is needed to find an adequate method for general 
higher dimensional problems. For some networks with special topology
 (such as the one in \fref{f:carre}(a)  but not the ones in \fref{f:carre}(b-d)), simple conditions
can be found for controllability, even for non-polynomial potentials (see \cite{EHR}). 
 \end{remark}

\section{Comparison with other commutator techniques}
It is perhaps useful to compare the techniques used in this paper to
those used elsewhere: To unify notation, we consider the
hypoellipticity problem in the classical form
\begin{equ}
L=  X_0+\sum_{i>0} X_i^2~.
\end{equ}
In \cite{EPR1999}, the authors considered a chain, so that $\VV_*$ is
just the first and the last particle in the chain. Starting with
$\partial _{p_1}$ (the left end of the chain) one then forms (with
simplified notation, which glosses over details which can be found in
that paper)
\begin{equ}
  \partial _{q_1}=[\partial _{p_1},X_0]~,\quad
  \partial _{p_2}=(M_{1,2})^{-1}[\partial _{q_1},X_0]~,\quad
  \partial _{q_2}=[\partial _{p_2},X_0]~,
\end{equ}
and so on, going through the chain. Here, the particles are allowed
to move in several dimensions, and $M_{j,j+1}$ is  basically
the Hessian matrix of $V_{j, j+1}$. This technique
requires that $M_{j,j+1}$ be invertible, which implies some restrictions
on the potentials.

Villani \cite{Villani2009} uses another sequence of commutators:
\begin{equ}
  C_0=\{X_i\}_{i>0} ~,\quad C_{j+1} = [C_j,X_0] + \text{ remainder$_j$}~.
\end{equ}

With this superficial notation, the current paper uses again a walk
through the network, but the basic step involves double commutators of
the form
\begin{equ}
  \scomm{Z_1}{Z_2}
\end{equ}
with $Z_i$ typically of the form $\sum g_v(x_v)\partial_{p_v}$, where we use
abundantly that the $V_e$ are polynomials. This allows for
the ``fanning out'' of \fref{f:61} and is at the basis of our ability
to control very general networks. In particular, this shows that
networks with variable cross-section can be controlled.

\appendix
\section{Vandermonde determinants}
\begin{lemma}\label{l:vandermonde} Let $c_1, \dots, c_n\in \real$ be
distinct and non-zero, and let $s \geq 0$.
 Then, for all $k\in \{1, \dots, n\}$ there are constants
 $r_{1}, \dots, r_n \in \real$ such that for all $j=1, \dots, n$,
\begin{equ}
\sum_{i=1}^{n} r_i\ c_j^{i+s} = \delta_{jk}~.
\end{equ}
\end{lemma}

\begin{proof}
 We have that the Vandermonde determinant
\begin{equ}
\left|\begin{array}{llll} c_1^{s+1} & c_1^{s+2} & \cdots & c_1^{s+n}\cr
 c_2^{s+1} & c_2^{s+2} & \cdots & c_2^{s+n}\cr
  & \vdots &  & \vdots \cr
 c_n^{s+1} & c_n^{s+2} & \cdots & c_n^{s+n}
\end{array}\right| = \left(\prod_{i=1}^n c_i^{s+1}\right)
\left | \begin{array}{lllll} 1& c_1 & c_1^2 & \cdots & c_1^{n-1}\cr
1 & c_2 &c_2^2& \cdots & c_2^{n-1}\cr
{}  & \vdots & \vdots &{} & \vdots \cr
 1 & c_n & c_n^2 & \cdots & c_n^{n-1}
\end{array}\right| =
\prod_{i=1}^n c_i^{s+1}\prod_{j=i+1}^n (c_j-c_i)
\end{equ}
 is non-zero under our assumptions. Thus, the columns of this
matrix form a basis of $\real^n$, which proves the lemma.
\qed\end{proof}

\acknowledgement
We thank Ch.~Boeckle, J. Guillod, T. Yarmola, and M. Younan for
discussions and careful reading of the manuscript.
This research was supported by an ERC advanced grant ``Bridges'' and
the Fonds National Suisse.
\endacknowledgement

\bibliography{networks.bbl}

\def\Rom#1{\uppercase\expandafter{\romannumeral #1}}\def\u#1{{\accent"15
  #1}}\def\cprime{$'$} \def\cprime{$'$}
\begin{thebibliography}{10}

\bibitem{EH}
J.-P. Eckmann and M.~Hairer.
\newblock Non-equilibrium statistical mechanics of strongly anharmonic chains
  of oscillators.
\newblock {\em {Commun.} {Math.} {Phys.}\/} {\bf 212} (2000), 105--164.

\bibitem{EHR}
J.-P. Eckmann, M.~Hairer, and L.~Rey-Bellet.
\newblock Non-equilibrium steady states for networks of springs.
\newblock {\em In preparation\/} .

\bibitem{EPR1999b}
J.-P. Eckmann, C.-A. Pillet, and L.~Rey-Bellet.
\newblock Entropy production in nonlinear, thermally driven {H}amiltonian
  systems.
\newblock {\em J. Statist. Phys.\/} {\bf 95} (1999), 305--331.

\bibitem{EPR1999}
J.-P. Eckmann, C.-A. Pillet, and L.~Rey-Bellet.
\newblock Non-equilibrium statistical mechanics of anharmonic chains coupled to
  two heat baths at different temperatures.
\newblock {\em Comm. Math. Phys.\/} {\bf 201} (1999), 657--697.

\bibitem{EZ2004}
J.-P. Eckmann and E.~Zabey.
\newblock Strange heat flux in (an)harmonic networks.
\newblock {\em J. Statist. Phys.\/} {\bf 114} (2004), 515--523.

\bibitem{hairer_probabilistic_2005}
M.~Hairer.
\newblock A probabilistic argument for the controllability of conservative
  systems.
\newblock {\em Arxiv preprint math-ph/0506064\/} .

\bibitem{Hairer2masses2009}
M.~Hairer.
\newblock How hot can a heat bath get?
\newblock {\em Comm. Math. Phys.\/} {\bf 292} (2009), 131--177.

\bibitem{Ho}
L.~H{\"o}rmander.
\newblock {\em The Analysis of Linear Partial Differential Operators {\Rom
  1--\Rom 4}\/} ({New York}: {Springer}, 1985).

\bibitem{jurdjevic_geometric_1997}
V.~Jurdjevic.
\newblock {\em Geometric control theory\/} (Cambridge; New York, {NY}, {USA}:
  Cambridge University Press, 1997).

\bibitem{bellet_ergodic_2006}
L.~Rey-Bellet.
\newblock Ergodic properties of {M}arkov processes.
\newblock {\em Open Quantum Systems {II}\/}  (2006), 1–39.

\bibitem{LTh00}
L.~Rey-Bellet and L.~E. Thomas.
\newblock Asymptotic behavior of thermal nonequilibrium steady states for a
  driven chain of anharmonic oscillators.
\newblock {\em Comm. Math. Phys.\/} {\bf 215} (2000), 1--24.

\bibitem{Villani2009}
C.~Villani.
\newblock Hypocoercivity.
\newblock {\em Mem. Amer. Math. Soc.\/} {\bf 202} (2009), iv+141.

\end{thebibliography}


\end{document}